\pdfoutput=1
\documentclass[runningheads]{llncs}
\usepackage{amssymb}
\usepackage{amsmath}
\let\doendproof\endproof
\renewcommand\endproof{~\hfill\qed\doendproof}
\usepackage{algorithm} 
\usepackage[noend]{algpseudocode} 
\usepackage{color}

\RequirePackage{doi}
\usepackage{hyperref}
\usepackage{graphicx}
\usepackage[margin,draft]{fixme}
\usepackage{xspace}

\usepackage{listings}

\usepackage{tabularx,booktabs}

\usepackage{pifont}
\newcommand{\cmark}{\ding{51}}%
\newcommand{\xmark}{\ding{55}}%

\newcommand{\set}[1]{\{#1\}}

\newcommand{\ie}{i.e.}
\newcommand{\eg}{e.g.}
\newcommand{\etc}{etc.}

\newcommand{\etal}{et al.}

\newcommand{\dbm}{DBM}
\newcommand{\dbms}{DBMs}
\newcommand{\DBMs}{Difference Bound Matrices}

\renewcommand{\implies}{\longrightarrow}

\renewcommand{\iff}{\longleftrightarrow}

\newcommand{\ta}{TA}
\newcommand{\tsle}{\preceq}
\newcommand{\tsse}{\trianglelefteq}
\newcommand{\TS}{(S, \to)}
\renewcommand{\ts}{(S, \to, \tsle)}
\newcommand{\reaches}[2]{#1 \to^* #2}
\newcommand{\reachesp}[2]{#1 \to^+ #2}
\newcommand{\steps}[3]{#1 \to #2 \to \ldots \to #3}
\newcommand{\run}[2]{#1 \to #2 \to \ldots}
\newcommand{\stepsprime}[3]{#1 \to' #2 \to' \ldots \to' #3}

\newcommand{\tole}{\to_{\tsle}}

\newcommand{\stepsle}[3]{#1 \tole #2 \ldots \tole #3}
\newcommand{\runle}[2]{#1 \tole #2 \tole \ldots}
\newcommand{\tose}{\to_{\tsse}}

\newcommand{\stepsse}[3]{#1 \tose #2 \tose \ldots \tose #3}

\newcommand{\rI}{(I, \tsse)}
\newcommand{\fw}{topological numbering}
\newcommand{\tosep}{\leadsto_{\tsse}}
\renewcommand{\sim}{\sqsubseteq}
\newcommand{\Sim}{\sim}
\newcommand{\zgto}{\Rightarrow}
\newcommand{\extra}{\alpha}
\newcommand{\zgtoalpha}{\Rightarrow_{\alpha}}
\newcommand{\dbmto}{\zgto_{\textit{DBM}}}
\newcommand{\alphalu}{\alpha_{\preceq LU}}
\newcommand{\alphasim}{\alpha_{\preceq}}
\newcommand{\tasim}{\preceq}
\newcommand{\simlu}{\tasim_{LU}}
\newcommand{\tastep}[5]{#2 \to_{{#3},{#4}} #5}
\newcommand{\zgalpha}{\zgtoalpha}
\newcommand{\zgalphasim}{\Rightarrow_{\alphasim}}
\newcommand{\Rpos}{\ensuremath{\mathbb{R}_{\ge 0}}}
\newcommand{\Nat}{\ensuremath{\mathbb{N}}}
\newcommand{\Integers}{\ensuremath{\mathbb{Z}}}

\newcommand{\vali}{\ensuremath{\mathbf{0}}}

\newcommand{\abstractionLU}{\alphalu}
\newcommand{\tosubsum}{\ensuremath{\rightsquigarrow}}
\newcommand{\extraLUp}{\ensuremath{\mathsf{Extra}_{LU}^+}}

\newcommand{\Y}{\cmark} 
\newcommand{\N}{\xmark}

\newcommand{\IfNoThen}[1]{\State\textbf{if} {#1}}
\newcommand{\ThenIndent}[1]{\State \hspace{\algorithmicindent}\textbf{then} #1}
\newcommand{\rejectcert}{\ThenIndent{\textbf{reject certificate}}}
\newcommand{\acceptcert}{\State\textbf{accept certificate}}

\usepackage{tikz}
\usetikzlibrary{arrows,automata}
\usetikzlibrary{decorations.pathmorphing}
\tikzstyle{every node}=[initial text=]
\tikzstyle{location}=[rounded corners, minimum size=12pt, draw=black, inner sep=2.5pt]

\usepackage{subfig}

\begin{document}
%
\title{Certifying Emptiness of Timed B\"uchi Automata}
%
%
\author{Simon Wimmer\inst{1}\orcidID{0000-0001-5998-4655}
\and
Frédéric Herbreteau\inst{2}\orcidID{0000-0002-1029-2356}
\and
Jaco van de Pol\inst{3}\orcidID{0000-0003-4305-0625} 
}
%
\authorrunning{S. Wimmer et al.}
%
\institute{%
Fakultät f\"ur Informatik, Technische Universität M\"unchen, Munich, Germany
\email{wimmers@in.tum.de}
\and
Univ. Bordeaux, CNRS, Bordeaux INP, LaBRI, UMR 5800, 33400, Talence, France 
\and
Aarhus University, Department of Computer Science, Denmark
}
%
\maketitle              
\begin{abstract}
Model checkers for timed automata are widely used to verify safety-critical, real-time systems.
State-of-the-art tools achieve scalability by intricate abstractions.
We aim at further increasing the trust in their verification results, in particular for checking liveness properties. 
To this end, we develop an approach for extracting certificates for the emptiness
of timed B\"uchi automata from model checking runs. These
certificates can be double checked by a certifier that we formally verify
in Isabelle/HOL. We study liveness certificates in an abstract setting 
and show that our approach is sound and complete.
To also demonstrate its feasibility,
we extract certificates for several models checked by TChecker and Imitator,
and validate them with our verified certifier.

\end{abstract}
\section{Introduction}
Real-time systems are notoriously hard to analyze due to intricate
timing constraints. A number of model checkers for timed automata (TA) \cite{Alur1994}
have been implemented and successfully applied to the verification of safety-critical timed
systems. Checking liveness properties of timed automata has revealed to be particularly
important, as emphasized by a bug in the standard model of the CSMA/CD protocol that has
been discovered only recently~\cite{Herbreteau2016}. Several algorithms have been implemented
to scale the verification of liveness specifications to larger 
systems~\cite{Tripakis:FMSD:2005,Tripakis:TOCL:2009,Li:FORMATS:2009,multicore_ta_ltl,Herbreteau2016}.
Users of timed automata model checkers put a high amount of trust in their verification results.
However, as verification algorithms
get more complex, it becomes highly desirable 
to justify the users' confidence in their correctness.

There are two main approaches to ensure high degrees of trustworthiness of automated tools:
verification and certification.
In the first approach, correctness of the verification tool (its implementation and its theory)
is proved using another semi-automated method.
This technique has been applied to model checkers \cite{munta,CAVA:2013} and SAT solvers \cite{VerifiedSAT}.
In the second approach, the automated tool produces a certificate, i.e.\ a proof for its 
verification result.
Then an independent tool, the certifier, checks that the proof is indeed valid.
In the best case, the certifier itself is formally verified.
Examples include SAT certificate checking \cite{GRAT,LRAT} and unreachability checking
of TA \cite{Tacas20}.

The certification approach promises many advantages over verification,
since certificate checking is much simpler than producing the certificate.
This drastically reduces the burden of semi-automated verification,
which is a laborious task. While proving correctness of a competitive verification 
tool might be prohibitively complicated, it may be feasible
to construct an efficient verified certifier instead (in the case of SAT \cite{GRAT},
the verified certifier was even faster than the original SAT solvers).
Finally, there is a wide variety of model checking algorithms 
and high-performance implementations, which are
suited for different situations. Instead of verifying them one by one,
these tools could produce certificates in a common format,
so they can be checked by a single verified certifier.

\subsection{Related Work}
Model checking LTL properties for timed automata 
~\cite{Alur1994,Tripakis:FMSD:2005,Tripakis:TOCL:2009,Li:FORMATS:2009,multicore_ta_ltl,Herbreteau2016}
consists of three conceptual steps: the LTL formula is transformed into a B\"uchi automaton, the semantics of the TA is computed as a (finite) zone graph, and the cross-product of these objects is checked for accepting cycles.
The two main alternative algorithms for detecting accepting cycles are
Nested Depth-First Search (NDFS) and the inspection of the Strongly Connected Components (SCC).
The NDFS algorithm was generalized to TA in LTSmin~\cite{multicore_ta_ltl,DBLP:conf/tacas/KantLMPBD15} and
extended to parametric TA in Imitator~\cite{DBLP:conf/iceccs/NguyenPP18,imitator}.
The SCC-based algorithm has also been generalized to TA in TChecker~\cite{Herbreteau2016,Tchecker}.
Both algorithms support {\em abstraction and subsumption between states} to reduce the state space.

{\em Verified model checking.}
An early approach targeted the verification of a $\mu$-calculus model checker in Coq~\cite{mu_calculus_certification}.
The NDFS algorithm was checked in the program verifier Dafny~\cite{DBLP:conf/fmics/Pol15,DBLP:conf/icse/Leino04},
while a multi-core version of it was checked in the program verifier Vercors~\cite{DBLP:conf/tacas/OortwijnHJP20,DBLP:conf/ifm/BlomDHO17}. 
A complete, {\em executable} LTL model checker was verified in the interactive theorem prover Isabelle/HOL~\cite{CAVA:2013}
and later extended with partial-order reduction \cite{Brunner2017}. 
A verified model checker for TA, Munta~\cite{munta}, has also been constructed in Isabelle/HOL~\cite{Tacas18,nipkow_isabelle/hol_2002}.

{\em Certification.}
A certifier for reachability properties in TA has been proposed very recently~\cite{Tacas20}. A certification approach for LTL model checking was proposed in~\cite{Griggio2018}. It uses k-liveness to reduce the problem to IC3-like invariant checking. 

{\em Contributions.}
In this paper, we extend certificates for unreachability of TA \cite{Tacas20} to certificates
for liveness properties, i.e.\ emptiness of timed B\"uchi automata (TBA). We propose a common certification approach 
for tools using different algorithms and various abstractions \cite{Herbreteau2016,multicore_ta_ltl}.
These certificates can be much smaller than the original state space, due to the use of subsumption and abstraction.
The difficulty here is that a careless application of subsumption can introduce spurious accepting cycles.
Our new contributions are \footnote{An artifact containing our code and benchmarks is available on \href{https://doi.org/10.6084/m9.figshare.12620582.v1}{figshare} \cite{wimmer_herbreteau_van_de_pol_2020}.}:
\begin{itemize}
\item We introduce an abstract theory for certificates of B\"uchi emptiness, which can be instantiated for zone graphs of TBA with subsumptions. 
\item We developed a fully, mechanically verified certifier in Isabelle/HOL.
In particular, our certifier retains the ability to check certificates in parallel.
\item We show that the previous certifier for reachability and our extension to B\"uchi emptiness
are compatible with implicit abstraction techniques for TA.
\item We demonstrate feasibility by generating and checking certificates for two external model checkers, representing the NDFS and the SCC approach.
\end{itemize}
Note that checking counter-examples is easy in practice, but checking ``true'' model checking results is much harder. This is exactly what we address with certifying emptiness of TBA.
The main application would be to increase the confidence in safety-critical real-time applications, which have been verified with an existing model checker.
Another possible application of the certifier would be to facilitate a new model checking contest for liveness properties of TA. 
\section{Timed Automata and Model Checking}
\label{sec:ta-mc}
In this section, we set the stage for the rest of the paper by recapitulating
the basic notions of TA and summarizing the essential concepts
of TBA verification.

\subsection{Verification Problems for Timed Automata}

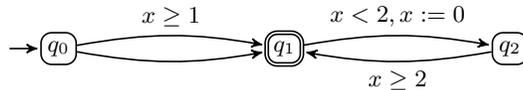
\begin{figure}[b]
    \centering 
\begin{tikzpicture}[->,>=stealth',shorten >=1pt,auto,node distance=3cm,semithick]
  \node[initial,location]	(A) {$q_0$};
  \node[location,accepting] (B) [right of=A] {$q_1$};
  \node[location] (C) [right of=B] {$q_2$};

  \path (A) edge [bend right=10] (B)
  					edge [bend left=10] node [above] {$x\geq 1$} (B)
          (B) edge [bend left=10] node [above] {$x<2, x:=0$} (C)
          (C) edge [bend left=10] node [below] {$x\geq 2$} (B);
\end{tikzpicture}
    \caption{\label{fig:ta}Timed (B\"uchi) automaton with initial state $q_0$ and accepting state $q_1$.}
\end{figure}

A TA $A=(Q,q_0,F,I,T,X)$ is a finite automaton extended with a finite set of {\em clocks} $X$. $Q$ is a finite set of states with initial state $q_0\in Q$ and accepting states $F \subseteq Q$. $I$ associates an {\em invariant} constraint to every state and 
$T$ associates a {\em guard} constraint $g$ and {\em clock reset} $R\subseteq X$ to each transition. 
Here \emph{(clock) constraints} are conjunctions of formulas $x \# c$, where $x$ is a clock, $c \in \Nat$ and $\# \in \set{<,\le,=,\ge,>}$.
Observe that we exclude diagonal constraints of the form $x-y\#c$.
An example of a timed automaton is depicted in Figure~\ref{fig:ta}.

A clock valuation $v : \, X \to \Rpos$ associates a non-negative real value to each clock $x \in X$. 
A configuration is a pair $(q,v)$ where $q$ is a state and $v$ is a clock valuation. The initial configuration is $(q_0,\vali)$.
Without loss of generality, we assume that the initial clock valuation $\vali$ satisfies the invariant $I(q_0)$. There are two kind of steps from a configuration $(q,v)$:
\begin{description}
    \item[delay] $(q,v) \to_{\delta} (q,v')$ for a delay $\delta \in \Rpos$ if for every clock $x \in X$, $v'(x) = v(x) + \delta$, and $v'$ satisfies the invariant $I(q)$;
    \item[transition] $(q,v) \to_t (q',v')$ for transition $t = (q,g,R,q') \in T$ if $v$ satisfies the guard $g$, $v'(x) = 0$ if $x \in R$ and $v'(x) = v(x)$ otherwise, and $v'$ satisfies $I(q')$.
\end{description}
We write $(q,v) \to_{\delta, t} (q',v')$ if there exists a configuration $(q,v'')$ such that $(q,v) \to_{\delta} (q,v'') \to_t (q',v')$. A run of a timed automaton is an (infinite) sequence of transitions of the form:
$(q_0,\vali) \to_{\delta_0,t_0} (q_1,v_1) \to_{\delta_1,t_1} \cdots$.
A run is \emph{non-Zeno} if the sum of its delays is unbounded.

The \emph{reachability problem} asks, given a timed automaton $A$, if there exists a finite run from the initial configuration $(q_0,\vali)$ to an accepting configuration $(q_n,v_n)$ such that $q_n \in F$.

In timed B\"uchi automata (TBA), $F$ is interpreted as a B\"uchi acceptance condition.
The \emph{liveness problem} then asks, whether a given TBA $A$ is non-empty, \ie\ if there is an infinite non-Zeno run from the initial configuration $(q_0,\vali)$ that visits infinitely many accepting configurations $(q_i,v_i)$ with $q_i \in F$. In this paper, we work under the common assumption that TA only admit non-Zeno runs (see \cite{Tripakis:FMSD:2005} for a construction to enforce this on every \ta).

Both problems are known to be PSPACE-complete~\cite{Alur1994}. Due to density of time, these two verification problems cannot be solved directly from the transition system induced by configurations and steps. A well-known solution to this problem is the region graph construction of Alur and Dill~\cite{Alur1994}.
Yet, it is not used in practice, as the region graph is enormous even for rather simple automata.

\subsection{Zone Graph and Abstractions}

The practical solution that is implemented in state-of-the-art tools like UPPAAL~\cite{Larsen1997}, TChecker~\cite{Tchecker} and the Imitator tool~\cite{imitator} is based on zones. Let us fix a set of clocks $X$. A zone $Z$ is a set of valuations represented as a conjunction of constraints of the form $x \# c$ or $x - y \# c$ for $x, y \in X$, $\# \in \set{<,\le,=,\ge,>}$ and $c \in \Integers$. Zones can be efficiently represented using \DBMs\ (\dbms)~\cite{DBM}. Moreover, zones admit a canonical representation, hence equality and inclusion of two zones can be checked efficiently~\cite{uppaal}.

We now define the symbolic semantics~\cite{Daws:TACAS:1998} of a TA $A$.
Let $q, q'$ be two states of $A$, and let $W, W' \subseteq \Rpos^X$ be two non-empty sets of clock valuations. We have $(q,W) \zgto^t (q',W')$ for some transition $t \in T$, if $W'$ is the set of all clock valuations $v'$ for which there exists a valuation $v \in W$ and a delay $\delta \in \Rpos$ such that $(q,v) \to_{\delta,t} (q',v')$. In other words, $W'$ is the strongest postcondition of $W$ along transition $t$. The symbolic semantics of $A$ , denoted by $\zgto$, is the union of all $\zgto^t$ over $t \in T$. The symbolic semantics is a sound and complete representation of the finite and infinite runs of $A$. Indeed, $A$ admits a finite (resp. infinite) run $(q_0,v_0) \to_{\delta_0,t_0} (q_1,v_1) \to_{\delta_1,t_1} \ldots (q_n,v_n) \to_{\delta_n,t_n} \ldots$ if and only if there exists a finite (resp. infinite) path $(q_0,W_0) \zgto^{t_0} (q_1,W_1) \zgto^{t_1} \ldots (q_n,W_n) \zgto^{t_n} \ldots$ such that $v_i \in W_i$ for all $i \ge 0$ and $W_0 = \set{\vali}$~\cite{Daws:TACAS:1998}.
It is well-known that if $Z$ is a zone, and $(q,Z) \zgto (q',W')$ then $W'$ is a zone as well~\cite{uppaal}. Since $\set{\vali}$ is a zone, all the reachable nodes in $\zgto$ are zones as well. The reachable part of $\zgto$ is called the \emph{zone graph} of $A$. The nodes of the zone graph are denoted as $(q,Z)$ in the sequel and the zone graph is simply denoted by its transition relation $\zgto$. Fig.~\ref{fig:zg1} depicts the zone graph of the automaton in Fig.~\ref{fig:ta}.

\begin{figure}[t]
    \subfloat[Zone Graph]{\centering 
\begin{tikzpicture}[->,>=stealth',shorten >=1pt,auto,node distance=1.3cm,semithick]
  \node[initial,location]	(A) {$q_0,x\geq 0$};
  \node[location,accepting] (B1) [below left of=A] {$q_1,x\geq 0$};
  \node[location,accepting] (B2) [below right of=A] {$q_1,x\geq 1$};
  \node[location] (C1) [below of=B1] {$q_2,x\geq 0$};
  \node[location] (C2) [below of=B2] {$q_2,x\geq 2$};
  \node[location,accepting] (B3)  [below right of=C1] {$q_1,x\geq 2$};
    
  \path (A) edge (B1)
  					edge (B2)
          (B1) edge (C1)
          (B2) edge (C2)
          (C1) edge (B3)
          (C2) edge (B3);
\end{tikzpicture}\label{fig:zg1}}
    \hfill
    \subfloat[Liveness compatible]{\centering 
\begin{tikzpicture}[->,>=stealth',shorten >=1pt,auto,node distance=1.3cm,semithick]
  \node[initial,location]	(A) {$q_0,x\geq 0$};
  \node[location,accepting] (B1) [below left of=A] {$q_1,x\geq 0$};
  \node[location,accepting] (B2) [below right of=A] {$q_1,x\geq 1$};
  \node[location] (C1) [below of=B1] {$q_2,x\geq 0$};
  \node[location,accepting] (B3)  [below right of=C1] {$q_1,x\geq 2$};
    
  \path (A) edge (B1)
  					edge (B2)
          (B1) edge (C1)
          (B2) edge [draw=blue,decoration={snake,amplitude=1pt,segment length=4pt},decorate] (B1)
          (C1) edge (B3);
\end{tikzpicture}\label{fig:zg2}}
    \hfill
    \subfloat[Not compatible]{\centering 
%

\begin{tikzpicture}[->,>=stealth',shorten >=1pt,auto,node distance=1.3cm,semithick]
  \node[initial,location]	(A) {$q_0,x\geq 0$};
  \node[location,accepting] (B1) [below left of=A] {$q_1,x\geq 0$};
  \node[location,accepting] (B2) [below right of=A] {$q_1,x\geq 1$};
  \node[location] (C1) [below of=B1] {$q_2,x\geq 0$};
  \node[location,accepting] (B3)  [below right of=C1] {$q_1,x\geq 2$};
    
  \path (A) edge (B1)
  					edge (B2)
          (B1) edge (C1)
          (B2) edge [draw=blue,decoration={snake,amplitude=1pt,segment length=4pt},decorate] (B1)
          (C1) edge (B3)
          (B3) edge [bend right=30,draw=blue,decoration={snake,amplitude=1pt,segment length=4pt},decorate] (B1);
\end{tikzpicture}\label{fig:zg3}}
    \caption{\label{fig:zgs}Three subsumption graphs for the automaton in Fig.~\ref{fig:ta}.}
\end{figure}
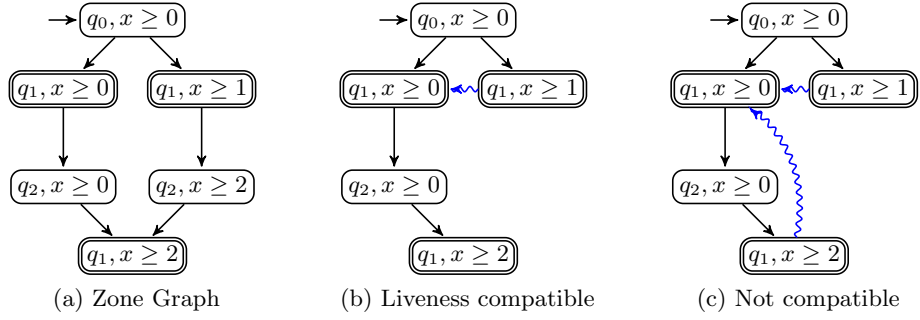

Still, the zone graph of a timed automaton may be infinite. As a remedy, finite abstractions have been introduced in the literature~\cite{Daws:TACAS:1998,static-guard-analysis,Behrmann:STTT:2006}.

An \emph{abstraction} $\extra$ transforms a zone $Z$ into a zone $\extra(Z)$ such that $Z \subseteq \extra(Z)$, $\extra(\extra(Z)) = \extra(Z)$, and every run that is feasible from a valuation $v' \in \extra(Z)$ is simulated by a run from a valuation $v \in Z$.
Such abstractions are called extrapolations in the literature~\cite{Behrmann:STTT:2006}.
An abstraction is finite when the set of abstracted zones $\set{\extra(Z) \, | \, Z \text{ is a zone}}$ is finite. Given an abstraction $\extra$, the \emph{abstracted zone graph} has initial node $(q,\extra(\set{\vali}))$ and transitions of the form $(q,Z) \zgtoalpha^t (q',\extra(Z'))$ for each transition $(q,Z) \zgto^t (q',Z')$. Let $\zgtoalpha$ denote the union of all $\zgtoalpha^t$ over $t \in T$. The abstracted zone graph is sound and complete: there is a run $(q_0,v_0) \to_{\delta_0,t_0} (q_1,v_1) \to_{\delta_1,t_1} \ldots\, (q_n,v_n)\; (\to_{\delta_n,t_n} \ldots)$ in $A$ if and only if there is an infinite path $(q_0,Z_0) \zgtoalpha^{t_0} (q_1,Z_1) \zgtoalpha^{t_1} \ldots\, (q_n,Z_n)\; (\zgtoalpha^{t_n} \ldots)$ with $v_i \in Z_i$ for all $i \ge 0$. Hence, when $\extra$ is a finite abstraction, the verification problems for a TA $A$ can be algorithmically solved from its abstracted zone graph. The abstraction $\extraLUp$~\cite{Behrmann:STTT:2006} is implemented by state-of-the-art verification tools UPPAAL~\cite{Larsen1997} and TChecker~\cite{Tchecker}. Our results hold for any finite, sound and complete abstraction. The abstracted zone graph is denoted $\zgtoalpha$ in the sequel.

\subsection{Subsumption}

Consider the TA in Figure~\ref{fig:ta} and its zone graph in Figure~\ref{fig:zg1}. Observe that every run that is feasible from node $(q_1, x \ge 1)$ is also feasible from $(q_1, x \ge 0)$ since the zone $x \ge 1$ is included in the zone $x \ge 0$ (recall that zones are sets of clock valuations). We say that $(q_1, x \ge 1)$ is \emph{subsumed} by the node $(q_1, x \ge 0)$. As a result, if an accepting node is (repeatedly) reachable from $(q_1, x \ge 1)$, then an accepting node is also (repeatedly) reachable from $(q_1, x \ge 0)$.

This leads to a crucial optimization for the verification of TA:
reachability and liveness verification problems can be solved without exploring subsumed nodes. This optimization is called inclusion abstraction in~\cite{Daws:TACAS:1998}. Figure~\ref{fig:zg2} shows the graph obtained when the exploration is stopped at node $(q_1, x \ge 1)$. All the runs that are feasible from $(q_1, x \ge 1)$ are still represented in this graph, as they can be obtained by first taking the subsumption edge from $(q_1, x \ge 1)$ to $(q_1, x \ge 0)$ (depicted as a blue squiggly arrow), and then any sequence of (actual or subsumption) edges from $(q_1, x \ge 0)$. Such graphs with both actual and subsumption edges are called \emph{subsumption graphs} in the sequel.

It is tempting to use subsumption as much as possible, and only explore maximal nodes (w.r.t.\ zone inclusion). While this is correct for the verification of reachability properties, subsumption must be used with care for liveness verification. The bottom node $(q_1, x \ge 2)$ in Figure~\ref{fig:zg2} is also subsumed by the node $(q_1, x \ge 0)$. A subsumption edge can thus be added between these two nodes as depicted in Figure~\ref{fig:zg3}. However, due to this new subsumption edge, the graph has a B\"uchi accepting path (of actual and subsumption edges) that does not correspond to any run of the timed automaton in Figure~\ref{fig:ta}. Indeed, subsumption leads to an overapproximation of the runs of the automaton. While all the runs from node $(q_1, x \ge 2)$ are feasible from node $(q_1, x \ge 0)$, the converse is not true: the transition $q_1 \xrightarrow{x<2, x:=0} q_2$ is not feasible from $(q_1, x \ge 2)$.

The subsumption graphs in Figure~\ref{fig:zg2} and~\ref{fig:zg3} can be seen as certificates issued by verification algorithms. The graph in Figure~\ref{fig:zg2} is a valid certificate for liveness verification as 1)~it contains no accepting paths, and 2)~every run of the automaton is represented in the graph. In constrast, the graph in Figure~\ref{fig:zg3} is not a valid certificate for liveness verification as it has an accepting path that does not correspond to any run of the automaton. In the next sections, we introduce an algorithm to check the validity of certificates produced by liveness verification algorithms, as well as a proven implementation of the algorithm.

\section{Certificates for B\"uchi Emptiness}
\label{sec:certificates_buchi_emptiness}

In this section, we study certificates for B\"uchi emptiness in the
setting of a slight variation of well-structured transition systems \cite{WSTS}.
First, we present reachability invariants, which certify that every run in the original system can be simulated on the states given in the invariant.
Next, we show that the absence of certain cycles in the invariant is sufficient
to prove that the original transition system does not contain accepting runs.
Then, we add a proof of absence of these cycles to the certificate.
Finally, we instantiate this framework for the case of \ta.

\subsection{Self-Simulating Transition Systems}
A \emph{transition system} $\TS$ consists of a set of states $S$ and a transition relation $\to \,\subseteq S \times S$.
If $S$ is clear from the context, we simply write $\to$.
We say that $\steps{s_1}{s_2}{s_n}$ is a path or that $\run{s_1}{s_2}$ is an (infinite) run
in $\to$ if $s_i \to s_{i+1}$ for all $i$.
Given an initial state $s_0$ and a predicate for accepting states $\phi$,
the path $\steps{s_0}{s_1}{s_n}$ is accepting if $\phi(s_n)$.
A run $\run{s_0}{s_1}$ is an (accepting) B\"uchi run if $\phi(s_i)$ for infinitely many $i$.

A transition system $\to$ is simulated by the transition system
$\to'$ if there exists a simulation relation $\sim$
such that:
\[
    \forall s, s', t.\; s \sim s' \wedge s \to t \implies \exists t'.\;
    s' \to' t' \wedge t \sim t'
\]
This \emph{simulation property} can be lifted to paths and runs:
\begin{proposition}
    \label{thm:simulation}
    If $\steps{s_1}{s_2}{s_n}\,(\to \ldots)$ is a path (run) and $s_1 \sim t_1$,
    then there is a path (run) $\stepsprime{t_1}{t_2}{t_n}\,(\to \ldots)$ with $s_i \sim t_i$
    for all $i$.
\end{proposition}

\begin{definition}
    A \emph{self-simulating transition system (SSTS)} $\ts$ consists of a transition system $\TS$
    and a quasi-order (a reflexive and transitive relation) $\tsle \,\subseteq S \times S$ on states such that $\to$ is simulated by $\to$ itself for $\tsle$.
\end{definition}
In comparison to well-structured transition systems \cite{WSTS},
our definition is slightly more relaxed, as we only demand that $\tsle$ is a quasi order,
not a well-quasi order.
Intuitively, transitivity of $\tsle$ is needed to allow for correct simulation by arbitrary ``bigger'' nodes.
In \ta, $\tsle$ corresponds to subsumption $\subseteq$, and $\to$ corresponds to $\zgto$.


\subsection{Reachability Invariants on Abstract Transition Systems}
In this section, we introduce the concept of \emph{reachability invariants}
for SSTS.
\begin{definition}
    A set $I \subseteq S$ is a \emph{reachability invariant} of an SSTS $\ts$
    iff for all $s \in I$ and $t$ with $s \to t$, there exists a $t' \in I$
    such that $t \tsle t'$.
\end{definition}
A useful invariant is also fulfilled by some inital state.
Such states will show up in theorems below.
In the remainder, unless noted otherwise, $\ts$ is an SSTS and $I$ is a reachability
invariant of it. Figures \ref{fig:zg1} to \ref{fig:zg3}
all form a reachability invariant for the zone graph from Figure \ref{fig:zg1}.

As was observed by Wimmer and von Mutius \cite{Tacas20}, reachability invariants can directly be applied as certificates for \emph{unreachability}.

\begin{definition}
    A predicate $\phi$ (for accepting states) is compatible with an SSTS $\ts$ iff
    for all $s, s' \in S$, if $\phi(s)$ and $s \tsle s'$, then also $\phi(s')$.
\end{definition}
An invariant $I$ can now certify that no accepting state $s$ with $\phi(s)$ is reachable:
\begin{theorem}
    \label{thm:reachability-certificate}
    If $\forall s \in I.\; \neg \phi(s)$, for some compatible $\phi$,
    $s_0 \in S$ and $s_0' \in I$ with $s_0 \tsle s_0'$, then
    there is no accepting path $\steps{s_0}{s_1}{s_n}$ with $\phi(s_n)$.
\end{theorem}

Note that this approach to certifying unreachability is also complete:
if no accepting state is reachable from $s_0$ in $\ts$, we can simply set $I := S$.
However, this is not practical for infinite transition systems, of course.
Thus we will revisit the question of completeness for \ta\ below.

Finally, we observe that the invariant
can be limited to a restriction of $\tsle$. 
\begin{definition}
    \label{def:restricted-reachability-invariant}
    A pair $\rI$ of a set $I \subseteq S$ and a binary relation $\tsse$ is a \emph{restricted reachability invariant} of an SSTS $\ts$
    iff:
    \begin{enumerate}
        \item For all $s \in I$ and $t$ with $s \to t$, there exists a $t' \in I$
        such that $t \tsse t'$.
        \item For all $s$, $t$, if $s \tsse t$, then also $s \tsle t$.
    \end{enumerate}
\end{definition}
In Figure\ 2, the $\tosubsum$-arrows would play the role of $\tsse$.
In Figure\ \ref{fig:zg2}, $(q_1, x \geq 0)$ is subsumed by both $(q_1, x \geq 1)$ and $(q_1, x \geq 2)$, but as we have seen in Figure\ \ref{fig:zg3}, it is crucial to disregard these subsumptions.
Therefore we need to consider restricted reachability invariants.

For any restricted reachability invariant, we can define a simulating transition system $\tose$:
\begin{definition}
    The transition system $(S, \tose)$ is defined such that
    $s \tose t'$ iff there exists a $t$ such that $s \to t$ and $t \tsse t'$.
\end{definition}
This simulation theorem is the key property of restricted reachability invariants \footnote{All proofs are omitted for brevity and can be found in the appendix.}:
\begin{theorem}
    \label{thm:inv-run-simulation}
    Given $s_1 \tsse t_1$ with $t_1 \in I$,
    if $\steps{s_1}{s_2}{s_n}\,(\to \ldots)$ is a path (run),
    then there is a path (run) $\stepsse{t_1}{t_2}{t_n}\,(\to \ldots)$ such that $s_i \tsse t_i$
    and $t_i \in I$ for all $i$.
\end{theorem}
Analogously to $\tose$, the transition system $\tole$ can be defined,
and Theorem \ref{thm:inv-run-simulation} can be proved for $\tole$.
This is used for the proof of Theorem \ref{thm:reachability-certificate} (see \cite{Tacas20}).

\subsection{B\"uchi Emptiness on Abstract Transition Systems}
\label{sec:buchi-emptiness-abstractTS}
In this section, we first give a general means of certifying that a transition system
does not contain a cycle, and then combine the idea with reachability invariants to
certify the absence of B\"uchi runs on SSTS.
\begin{definition}
    \label{def:forward-numbering}
    Given a transition system $\to$ and an accepting state predicate $\phi$,
    a \emph{\fw} of $\to$
    is a function $f$ with an integer range such that:
    \begin{enumerate}
        \item For all $s, t$, if $s \to t$, then $f(s) \geq f(t)$.
        \item For all $s, t$, if $s \to t$ and $\phi(s)$, then $f(s) > f(t)$.
    \end{enumerate}
\end{definition}
\begin{proposition}
    \label{thm:topo-numbering-no-accepting-cycle}
    Let $f$ be a topological numbering of $\to$ and $\phi$.
    If there exists a path of the form $\steps{s}{s_1 \to s_2}{s}$,
    then $\neg \phi(s)$.
\end{proposition}
These certificates are also complete:
\begin{proposition}
    \label{topo-numberings-complete}
    If there is no path $\steps{s}{s_1 \to s_2}{s}$ with $\phi(s)$ in $\to$,
    then the following are \fw s for $\to$.
    \begin{enumerate}
        \item The number of accepting states that are reachable from a node: $f(s) := |\{x \,|\, \reaches{s}{x} \wedge \phi(x)\}|$ (assuming $\{x \,|\, \reaches{s}{x} \wedge \phi(x)\}$ is finite for any $s$).
        \item If $h$ is a topological numbering (in the classical sense) of the strongly connected components (SCCs) of $\to$,
        then set $g(s) := h(C)$ if $s \in C$.
    \end{enumerate}
\end{proposition}
We now lift this idea to the case of (restricted) reachability invariants.
\begin{definition}
    \label{def:restricted-forward-numbering}
    Given an SSTS $\ts$, an accepting state predicate $\phi$,
    and a corresponding restricted reachability invariant $\rI$,
    a \emph{restricted topological numbering} of $\ts$
    is a function $f$ with an integer range such that:
    \begin{enumerate}
        \item For all $s, t' \in I$ and $t \in S$, if $s \to t$, and $t \tsse t'$, then $f(s) \geq f(t')$.
        \item For all $s, t' \in I$ and $t \in S$, if $s \to t$, $t \tsse t'$, and $\phi(s)$, then $f(s) > f(t')$.
    \end{enumerate}
    Moreover, let $\tosep$ be the restriction of $\tose$ to $I$, \ie\ the transition system such that $s \tosep t'$
    iff $s, t' \in I$ and there exists a $t$ such that $s \to t$ and $t \tsse t'$.
\end{definition}
Now, $f$ is clearly a \fw\ for $\tosep$. Thus $\tosep$ is free of accepting cycles.
Additionally, the transition system $\tosep$ trivially simulates $\tose$ with $s \Sim s'$ iff $s' = s$ and $s \in I$.
Therefore, any accepting cycle $s \tose^+ s$ in $\tose$ with $s \in I$ and $\phi(s)$ yields an accepting cycle $s \tosep^+ s$. Hence $\tose$ is free of accepting cycles.

From this, we conclude our main theorem that allows one to certify absence of B\"uchi runs in a transition system $\to$.
\begin{theorem}
    \label{thm:no-buechi-run}
    Let $f$ be a restricted \fw\ of $\ts$ for a compatible predicate $\phi$ and a finite restricted reachability invariant $\rI$.
    Then, for any initial state $s_0 \in S$ with $s_0 \tsse t_0$ for $t_0 \in I$,
    there is no B\"uchi run from $s_0$.
\end{theorem}
In practice, a certificate can now be given as a finite restricted reachability invariant $I$ as described above,
and a corresponding restricted topological numbering $f$. Both properties can be checked locally for each individual state
in $I$.

\subsection{Instantiation for Timed Automata}
We now want to instantiate this abstract certification framework for the concrete case
of TBA.
Our goal is to certify that the zone graph $\zgto$ does not contain any B\"uchi runs.
As the zone graph is complete, this implies that the underlying TBA is empty.
Thus we set $\to \;:=\; \zgto$.
Subsumptions in the zone graph shall correspond to the self-simulation relation of the SSTS.
Hence we define $\tsle$ such that $(q, Z) \tsle (q', Z')$ iff $q' = q$ and $Z \subseteq Z'$.

To certify unreachability, it is sufficient to consider arbitrary subsumptions in the zone graph,
\ie\ $\tsse \;:=\; \tsle$ \cite{Tacas20}. In other words it is sufficient to check that
the given certificate $I$ is a reachability invariant for $\ts$.
We have not yet given the set of states $S$.
Abstractly,
$S$ is simply the set of non-empty states,
\ie\ $S \;:=\, \{(q, Z) \,|\, Z \neq \emptyset\}$.
If it was allowed to reach empty zones, then soundness of the zone graph would not be given.
In practice, the certifier needs to be able to compute $\zgto$ effectively,
typically using the DBM representation of zones.
To this end one wants to add the assumption on states that
all DBMs are in canonical form.
One needs to ensure that
states are split according to $\phi$, \ie\ $\forall (q, Z) \in S.\, Z \subseteq \Phi(q) \vee Z \cap \Phi(q) = \emptyset$ where $\Phi(q) = \{v \,|\, \phi(q, v)\}$.
This is trivial for commonly used properties that concern only the finite state part.

Following these considerations, we propose the following certifier for the emptiness of TBA.
A certificate $C$ is a set of triplets $(q, Z, i)$ where $q$ is a discrete state,
$Z$ is a corresponding zone, and $i$ is the topological number for $(q, Z)$.
The certifier runs Algorithm \ref{alg:buechi-emptiness} on this certificate.
The algorithm extends the one by Wimmer and Mutius \cite{Tacas20} with
the topological numbers for liveness checking.

\begin{algorithm}
    \caption{Certifier for the emptiness of TBA}
\label{alg:buechi-emptiness}

                
\begin{algorithmic}[1]
    \Procedure{B\"uchi-Emptiness}{$\phi, C, q_0$}
    \ForAll {$(q, Z, i) \in C$}
    \Comment{\parbox[t]{.34\linewidth}{All DBMs are well-formed}}
        \IfNoThen{$Z = \emptyset \vee Z\; \text{is not canonical}$}
            \rejectcert
    \EndFor

    \IfNoThen {$\nexists (q_0, Z_0, i) \in C.\, \{\vali\} \subseteq Z_0$}
    \Comment{\parbox[t]{.34\linewidth}{The initial state is covered}}
        \rejectcert
    \ForAll {$(q, Z, i) \in C$}
    \Comment{\parbox[t]{.34\linewidth}{The certificate is:}}
        \ForAll {$(q_1, Z_1)$ s.t.\ $(q, Z) \zgto (q_1, Z_1)$}
            \IfNoThen{{$(\nexists (q_1, Z_1', j) \in C.\, Z_1 \subseteq Z_1'$
            \Comment{\parbox[t]{.34\linewidth}{an invariant,}}
                \par\hskip\algorithmicindent\hskip\algorithmicindent\hskip\algorithmicindent
                $\wedge\, (\phi(q) \implies i > j) \wedge i \geq j)$}}
                \Comment{\parbox[t]{.34\linewidth}{and a topological numbering}}
                \ThenIndent{\textbf{reject certificate}}
            
        \EndFor
    \EndFor
    \acceptcert
    \EndProcedure
\end{algorithmic}
\end{algorithm}

\begin{theorem}
    \label{thm:certificate-checker-correct}
    If $\textsc{B\"uchi-Emptiness}(\phi, C, q_0)$ accepts the certificate,
    then $\dbmto$ has no B\"uchi run for $\phi$. Consequently, the underlying TBA is empty.
\end{theorem}
The proof constructs a suitable $\tsse$ such that $(q,Z)\tsse(q,Z')$ if $Z\subseteq Z'$ and
$(q,Z',k)\in C$, where $k$ is selected to be minimal. 
Setting $I := \set{(q, Z) \,|\, \exists i.\, (q, Z, i) \in C)}$ and $f(q, Z) := \textsf{min}\{i\mid (q, Z, i) \in C\}$, Theorem~\ref{thm:no-buechi-run} can be applied.

The algorithm inherits several beneficial properties from \cite{Tacas20}.
First, it can easily be parallelized.
Most importantly however, the certifier does not need to compute an abstraction operation $\alpha$.
Suppose the model checker starts with a state $(q_0, \{\vali\})$ and explores the transition
$(q_0, \{\vali\}) \zgto (q_1, Z_1)$.
The model checker could then abstract zone $Z_1$  to $\alpha(Z_1)$,
and explore more edges from $(q_1, \alpha(Z_1))$,
\eg\ $(q_1, \alpha(Z_1)) \zgto (q_2, Z_2)$.
The certificate just needs to include $(q_0, \{\vali\})$, $(q_1, \alpha(Z_1))$,
and $(q_2, \alpha(Z_2))$, and the certificate checker just needs to check the following inclusions:
$\{\vali\} \subseteq \{\vali\}$, $Z_1 \subseteq \alpha(Z_1)$, and $Z_2 \subseteq \alpha(Z_2)$.
The checker does not need to compute $\alpha$
as $\alpha(Z_1)$ and $\alpha(Z_2)$ are part of the certificate.


It is rather easy to see that these certificates are also complete for timed automata.
For any finite abstraction $\extra$,
the abstracted zone graph $\zgtoalpha$ is finite and complete.
Thus, for a starting state $(q_0, \{\vali\})$ the set
\[I := \{(q, Z) \,|\, (q_0, \{\vali\}) \zgtoalpha^* (q, Z)\}\]
is a trivial finite reachability invariant that can be computed effectively
for common abstractions $\alpha$.
Moreover, if the underlying TBA is empty,
then $\zgtoalpha$ cannot contain a B\"uchi run either, since the abstract zone graph is complete.
Because $\zgtoalpha$ is finite, this means it cannot contain a cycle
through $\phi$.
Hence a forward numbering of $I$ can be given by computing the strongly connected
components of $I$.
However, this type of certificate is not of practical interest
as subsumptions are not considered. 
How certificates can be obtained for model
checking algorithms that make use of subsumption is the topic of Section
\ref{sec:algorithms}. 

\section{Incorporating Advanced Abstraction Techniques}
\label{sec:incorporating_state_of_the_art_abstraction}
We have already discussed that the techniques that were presented above are in principle agnostic
to the concrete abstraction $\alpha$ used.
This, however, is only true for standard verification algorithms for T(B)A
that use zone inclusion $Z \subseteq Z'$ as a simulation relation on the abstract zone graph.
There is also the noteworthy abstraction $\alphalu$ \cite{Behrmann:STTT:2006},
which is the coarsest zone abstraction that can be defined from clock bounds $L,U$~\cite{better-abstractions}.
Herbreteau \etal\ have shown that even though $\abstractionLU(Z)$ is usually not a zone,
it can be checked whether $Z \subseteq \abstractionLU(Z')$ directly from the \dbm\ representation of $Z$ and $Z'$,
without computing $\abstractionLU(Z')$~\cite{better-abstractions}. 
Hence, one can use $\abstractionLU$-subsumption over zones,
$Z \subseteq \abstractionLU(Z')$, instead of standard inclusion $Z \subseteq Z'$
to explore fewer symbolic states.
This technique can also be integrated with our certification approach.
This time, we will need more knowledge about the concrete abstraction $\alpha$, however.


We first describe the concept of time-abstract simulations, on which the definition
of $\alphalu$ is based.
\begin{definition}
    \label{def:time-abstract-simulation}
    A time-abstract simulation between clock valuations is a quasi-order $\tasim$ such that
    if $v \tasim v'$ and $\tastep{A}{(q, v)}{\delta}{t}{(q_1, v_1)}$ then
    there exist $\delta'$ and $v_1'$ such that
    $\tastep{A}{(q, v')}{{\delta'}}{t}{(q_1, v_1')} \wedge v_1 \tasim v_1'$.
\end{definition}
Behrmann \etal\ defined the simulation $\simlu$ based on the clock bounds $L$ and $U$,
and showed that it is a time-abstract simulation
\cite{Behrmann:STTT:2006}
(in fact one can show that $\simlu$ is even a simulation, \ie\ $\delta' = \delta$).
For any $\tasim$, one can define the corresponding abstraction
$\alphasim(Z) = \{v \,|\, \exists v' \in Z.\, v \tasim v'\}$.
This yields a sound and complete abstraction for any time-abstract simulation $\tasim$
\cite{Behrmann:STTT:2006}.
Observe that $\alphasim(Z)$ is the set of all valuations that are simulated by a valuation in $Z$ w.r.t. $\tasim$. As a result, every sequence of transitions feasible from $\alphasim(Z)$ is also feasible from $Z$ (although with different delays).

The implicit abstraction technique based on the
subsumption check $Z \subseteq \alphasim(Z')$ is compatible with our certification
approach for any $\alphasim$ for which $\tasim$ is a time-abstract simulation,
and in particular $\alphalu$.
Actually, we are still able to use algorithm \textsc{B\"uchi-Emptiness}
with the only modification that the condition $Z_1 \subseteq Z_1'$
is replaced with $Z_1 \subseteq \alphasim(Z_1')$.
We will justify this by showing that if the algorithm accepts the certificate,
then it represents a restricted reachability invariant with a suitable topological numbering
for $\zgalphasim$.
This means that $\zgalphasim$ does not have a B\"uchi run (Theorem \ref{thm:no-buechi-run}),
which, as $\alphasim$ is a complete abstraction, implies that the underlying TBA does not have a B\"uchi run either.

We first prove the following monotonicity property
(which can be seen as a generalization of Lemma 4 in the work of Herbreteau \etal\ \cite{better-abstractions}).
\begin{proposition}
    \label{thm:zg-mono}
    Let $\tasim$ be a time-abstract simulation.
    If $\alphasim(W) \subseteq \alphasim(W')$, $(q, W) \zgto^t (q_1, W_1)$, and $(q, W') \zgto^t (q_1, W_1')$,
    then $\alphasim(W_1) \subseteq \alphasim(W_1')$.
\end{proposition}
\newcommand{\siminv}{\sqsupseteq}
Reminding ourselves that $\alphasim$ is idempotent, if follows that if
$(q, W) \zgto^t (q_1, W_1)$ and $(q, \alphasim(W)) \zgto^t (q_1, W_1')$ for some states $q$, $q_1$, and sets of valuations $W$, $W_1$, and $W_1'$,
then $\alphasim(W_1) = \alphasim(W_1')$.
In other words, $\zgto$ simulates $\zgalphasim$ for $\siminv$ defined as $(q, W) \siminv (q, Z) \iff W = \alphasim(Z)$.

Now, we show that the conditions of Definitions \ref{def:restricted-reachability-invariant}
and \ref{def:restricted-forward-numbering} can be transferred along this simulation.
\begin{theorem}
    \label{thm:topo-numbering-simulation}
    Assume that the following conditions hold:
    \begin{enumerate}
        \item For all states $q$, and zones $Z$, $Z'$, $Z''$, if $(q, Z) \tsse (q, Z')$, then $Z \subseteq \alphasim(Z')$. Moreover, if $\alphasim(Z) = \alphasim(Z')$ and $(q, Z) \tsse (q, Z'')$,
        then $(q, Z') \tsse (q, Z'')$.
        \item For all $q$, $Z$, if $\phi((q, \alphasim(Z)))$, then $\phi((q, Z))$.
        \item $\rI$ satisfies condition (1) of Definition \ref{def:restricted-reachability-invariant}
        for $\zgto$.
        \item $f$ is a restricted topological numbering for $\zgto$, $\rI$, and $\phi$.
    \end{enumerate}
    Let $(q, W) \tsse' (q, W') \iff \exists Z, Z'.\, W = \alphasim(Z) \wedge W' = \alphasim(Z') \wedge (q, Z) \tsse (q, Z')$,
    $I' := \set{s' \,|\, \exists s \in I.\, s \sim s'}$
    and $f'(s') := \mathit{Min}\, \set{f(s) \,|\, s \in I \wedge s \sim s'}$. Then
    \begin{enumerate}
        \item $(I', \tsse')$ is a restricted reachability invariant for $(\zgalphasim, \subseteq)$.
        \item $f'$ is a restricted topological numbering for $\zgalphasim$, $(I', \tsse')$, and $\phi$.
    \end{enumerate}
\end{theorem}

Algorithm $\textsc{B\"uchi-Emptiness}$ ensures that there exist an invariant $\rI$
and a numbering $f$
that fulfill the conditions of Theorem \ref{thm:topo-numbering-simulation} for $\zgto$
(as indicated after Theorem \ref{thm:certificate-checker-correct}).
Thus, if the algorithm accepts the certificate,
there is a  restricted reachability invariant $(I', \tsse')$ with a corresponding topological
numbering $f'$ for $\zgalphasim$. Hence $\zgalphasim$ does not have a B\"uchi run.

\section{Evaluation}
\label{sec:evaluation}
In this section, we first give a brief description of the model checking algorithms
we consider and describe how certificates can be extracted from them.
Then, we outline the general architecture of our certification tool chain,
and finally we present some experiments on standard TA models.

\subsection{Extracting Certificates From Model Checkers}
\label{sec:algorithms}

We consider the two state-of-the-art algorithms for checking
B\"uchi emptiness for TA: the NDFS-based algorithm by Laarman \etal\ \cite{multicore_ta_ltl} and the iterative SCC-based algorithm by Herbreteau \etal\ \cite{Herbreteau2016}.
Both algorithms can be applied to any abstracted zone graph $\zgtoalpha$ for a finite, sound and complete
abstraction $\alpha$.
As was noted by Herbreteau \etal\ \cite{Herbreteau2016}, they also have in common that their correctness can be justified 
on the basis that they both compute subsumption graphs that are \emph{liveness compatible}, in the sense that they do not contain any cycle with an accepting node and a subsumption edge.

Considering NDFS for TA from~\cite{multicore_ta_ltl} more closely,
it prunes the search space by using subsumption in certain safe places. In particular,
the outer (blue) search is pruned when it reaches a state $s$ that is subsumed 
by a state on which the inner (red) search has been called, 
i.e.~$s\sqsubseteq t$ and $t$ is red. In order to generate a liveness-compatible subsumption
graph, the blue search exports all the states which are not subsumed along with their $\to$-successors. Moreover, 
the algorithm exports $\tosubsum$-edges as soon as the pruning by subsumption is applied. 

The iterative algorithm from~\cite{Herbreteau2016} interleaves reachability analysis and SCC decompositions. The reachability analysis computes a subsumption graph with maximal subsumption: a subsumption edge $s \tosubsum t'$
is added whenever a new state $t$ is visited from $s$, and $t$ is subsumed by some visited state $t'$. The resulting graph $\to \cup \tosubsum$ is a subsumption graph that preserves state reachability, but that may not be liveness compatible. Therefore, an SCC decomposition is run, and all subsumption edges from SCCs that contain both an accepting node and a subsumption edge are removed. States which are not subsumed anymore are re-explored in the next iteration of the main loop. 
Upon termination, the subsumption graph $\to \cup \tosubsum$ is liveness compatible.

Both algorithms compute liveness compatible subsumption graphs. In order to obtain a certificate we run one extra SCC decomposition of the graph with ${\to}\cup{\tosubsum}$-edges from which we compute a topological ordering.

\subsection{General Architecture}
\begin{figure}[b]
\includegraphics[width=\textwidth]{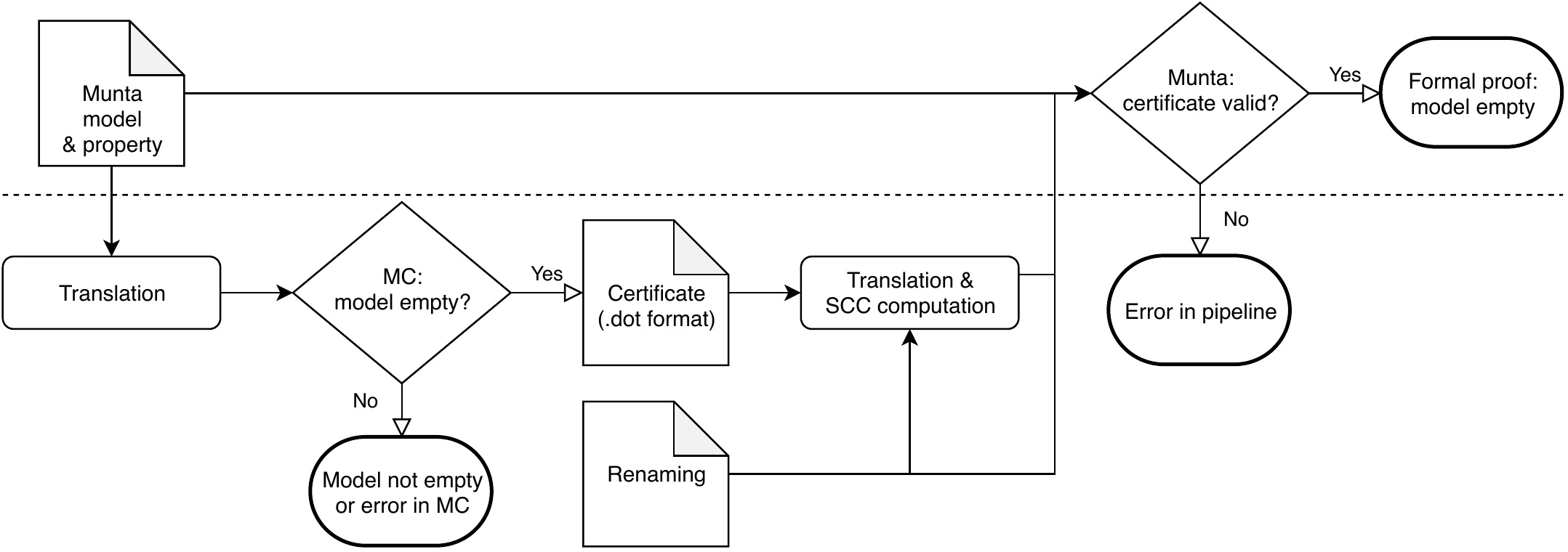}
\caption{Workflow of the certifier pipeline. The dashed line is the trust boundary.
If the correct model is given, then the answer on the right can be trusted.}
\label{fig:pipeline}
\end{figure}
Our certifier is implemented as an extension of the tool Munta \cite{munta},
which has been fully verified in Isabelle/HOL \cite{Tacas18,Tacas20}.
Figure \ref{fig:pipeline} depicts the architecture of our tool chain to certify the
emptiness of a given TBA.
The model (a TBA) and the acceptance property are given in the input format of Munta.
For the model checker in the middle, we used Imitator and TChecker.
In a first step, the Munta model is translated to an input model for
the model checker. 
The model checker decides whether the given TBA is empty.
If not, then either the model checker's answer is correct or it has found
a spurious counterexample; in both cases no certificate can be extracted.
Otherwise, the model checker emits a certificate consisting of a number of symbolic states and the set of edges in the subsumption graph.
The latter can either include 
proper transitions ($\to$) and subsumptions (\tosubsum) (this is done for Imitator with NDFS and subsumption),
or the edges that merge these two types (\tosubsum')
(which is done for TChecker and for Imitator with state merging enabled, see Section~\ref{sec:Experiments}).
In either case, in the next step where the certificate is translated to Munta's
binary input format for certificates, 
the SCC numbers (c.f.\, Proposition \ref{topo-numberings-complete}) are re-computed blindly
from these edges.
This step additionally makes use of a renaming dictionary to map from
human readable labels for states, actions, \etc, to natural numbers.

Finally, the TBA model, the translated certificate, and the renaming are given to Munta.
If it accepts the certificate, then there is an Isabelle/HOL theorem that
guarantees that the given TBA is indeed empty.
If the certificate is rejected, any of the steps in the tool chain
could have failed.
Note that the basis of trust 
is minimal.
One just needs to ensure that the model represents what one has in mind,
and to trust the correctness of Munta.
To trust Munta, one essentially needs to trust its TBA semantics,
which is less than 200 lines long, some core parts of Isabelle/HOL,
and an SML compiler (MLton in our case).
For details, we refer the interested reader to previous publications on Munta \cite{munta,Tacas20}.

\subsection{Experiments}
\label{sec:Experiments}

We have evaluated our approach on the TBA models that were also used
by Herbreteau \etal\ \cite{Herbreteau2016}.
These are inspired by standard TA benchmarks,
and all consist of the product of a TA model
and an additional B\"uchi automaton that encodes the complement
of the language of a given LTL formula that one wants to check.
Details are given by Herbreteau \etal\ \cite{Herbreteau2016}.

For Imitator we tried two methods: NDFS with subsumption and
reachability analyis with merging~\cite{DBLP:conf/rp/AndreS11}. 
Imitator does not apply abstractions
(since it was designed for parametric TA), so the full zone graph is often infinite
and most NDFS runs fail. The one that succeeds generates a valid certificate.
{\em Merging} tries to reduce the 
number of zones, by computing the exact convex hull of zones.
This creates new zones that could subsume several existing ones,
and often yields a finite zone graph. The certificate produced by merging
is always a reachability invariant but not necessarily a subsumption graph.
Merging may introduce spurious cycles,
in which case the certificate is not liveness compatible;
these cases are caught by the Munta certifier. If there are no (spurious) accepting
cycles, we obtain a valid and quite small certificate.
Note that the generalization from subsumption graphs to our certificates is crucial
to allow for merging.

Table \ref{tab:results} summarizes our experimental results.
TChecker was run with the algorithm from \cite{Herbreteau2016} and~\cite{multicore_ta_ltl}
and Imitator with the algorithm from~\cite{multicore_ta_ltl}, and with a reachability procedure with full merging.
The *** entries indicate cases where Imitator did not terminate
within 30s.
The results show that the certifier accepts those certificates that we
expect it to accept, but also rejects those that stem from subsumption graphs
that are not liveness compatible.
Moreover, the certifier was fully verified in Isabelle/HOL and still yields reasonable performance, 
certifying models with more than a 100k symbolic states in under 230s.
\begin{table}
\setlength{\tabcolsep}{3.5pt}
\begin{tabularx}{\textwidth}{l X crr X crr X crr X crr}
    \toprule
    Model & \multicolumn{7}{c}{TChecker} & & \multicolumn{7}{c}{Imitator} \\
    \cmidrule{3-9} \cmidrule{11-17}
    && \multicolumn{3}{c}{Iterative SCC} & & \multicolumn{3}{c}{NDFS} & & \multicolumn{3}{c}{Merge}  & & \multicolumn{3}{c}{NDFS} \\
    \midrule
    CC1 &  & \Y &     57 &   0.01 &  & \Y &   3281 &   0.06 &  & \Y &     58 &   0.01 &  & \multicolumn{3}{c}{***} \\
    CC4 &  & \Y & 195858 & 221.56 &  & \Y &  32575 &   7.75 &  & \multicolumn{3}{c}{***} &  & \multicolumn{3}{c}{***} \\
    CC5 &  & \Y &  65639 &  30.63 &  & \Y & 143057 & 218.98 &  & \multicolumn{3}{c}{***} &  & \multicolumn{3}{c}{***} \\
    FD1 &  & \Y &    214 &   0.02 &  & \Y &    677 &   0.03 &  & \N &    294 &   0.02 &  & \Y &   1518 &   0.11 \\
    FI1 &  & \Y &     65 &   0.01 &  & \Y &     71 &   0.00 &  & \Y &    136 &   0.00 &  & \multicolumn{3}{c}{***} \\
    FI2 &  & \Y &    314 &   0.01 &  & \Y &    344 &   0.01 &  & \Y &    589 &   0.01 &  & \multicolumn{3}{c}{***} \\
    FI4 &  & \Y &    204 &   0.00 &  & \Y &    224 &   0.01 &  & \Y &    793 &   0.01 &  & \multicolumn{3}{c}{***} \\
    FI5 &  & \Y &   3091 &   0.13 &  & \Y &   2392 &   0.09 &  & \N &    863 &   0.03 &  & \multicolumn{3}{c}{***}\\
    \bottomrule
\end{tabularx}
\medskip
\caption{Benchmark results on a 2017 MacBook Pro with 16 GB RAM and a Quad-Core Intel Core i7 CPU at 3.1 GHz.
For each algorithm, we show whether the certificate was accepted,
the number of DBMs in the certificate, and the time for certificate checking on a
single core in seconds.
}
\label{tab:results}
\end{table}

\section{Conclusion}
Starting from an abstract theory on self-simulating transition systems,
we have presented an approach to extract certificates from state-of-the-art model checking algorithms
(including state-of-the-art abstraction techniques)
that decide emptiness of timed B\"uchi automata. The certificates prove that a given model is indeed B\"uchi empty.
We have verified the theory and a checker for these certificates in Isabelle/HOL,
using the tool Munta as a basis. We demonstrated that our approach is feasible
by extracting certificates for some standard benchmark models from the tools TChecker and Imitator.
We hope that our work can help to increase confidence in safety-critical systems
that have been verified with timed automata model checkers.
Furthermore, we envision that our tool could help in the organization of future competitions
for such model checkers.

To close, we want to illuminate some potential future directions of research.
First, one is usually not only interested in the emptiness of TBA per se,
but more generally in the question if a TA model satisfies some LTL requirements.
Thus, 
our tool would ideally be combined
with a verified translation from LTL formulas to B\"uchi automata or with a certifier for such a construction.
The former has been realized by the CAVA project \cite{CAVA:2013}, while an avenue towards the latter
is opened by the recent work of Seidl \etal\ \cite{ltl-master-theorem-isabelle}.

Second, Herbreteau \etal\ have developed a technique of computing abstractions for TA
on the fly, starting from very coarse abstractions and refining them as needed
\cite{lazy-abstractions-ta}.
It seems that our approach is in principle compatible with this technique when augmenting
certificates with additional information on the computed abstractions,
whose validity would have to be checked by the certifier.

Third, one could attempt to reduce the size of the certificates.
In one approach, reachability certificates have been compressed
after model checking (c.f.\ \cite{Tacas20}).
On the other hand, model checking algorithms could speculate that the given TBA is empty,
and use this fact to use additional subsumptions to reduce the search space,
while risking to miss accepting runs. However, given the certification step afterwards,
this is of no concern. For instance, one could remove the red search
from the NDFS algorithm, and use subsumption on blue nodes instead of red nodes,
as a quick pre-check. If the result passes the certifier, we are done.

Finally, 
as our theory is not specific to timed automata per se,
it could be interesting to find other application domains for this approach to certification.
In light of the large body of existing work on well-structured transition systems,
this looks particularly promising as any such system is also self-simulating.

\bibliographystyle{splncs04}
\bibliography{references}

\clearpage\section*{Appendix}

The appendix restates all propositions of the paper together with their proofs.
For coherence, some propositions are stated slightly more precisely than in the paper.

\setcounter{lemma}{0}
\setcounter{proposition}{0}
\setcounter{theorem}{0}

\begin{proposition}
    \label{thm:simulation}
    If $\steps{s_1}{s_2}{s_n}\,(\to \ldots)$ is a path (run) and $s_1 \sim t_1$,
    then there is a path (run) $\stepsprime{t_1}{t_2}{t_n}\,(\to \ldots)$ with $s_i \sim t_i$
    for all $i$.
\end{proposition}

\begin{theorem}
    Let $I$ be an invariant and $\phi$ be a compatible predicate for an SSTS $\ts$.
    If $\forall s \in I.\; \neg \phi(s)$,
    then for all $s_0 \in S$ for which there exists an $s_0' \in I$ with $s_0 \tsle s_0'$,
    there is no accepting path $\steps{s_0}{s_1}{s_n}$ with $\phi(s_n)$.
\end{theorem}
\begin{proof}
    Assume that there is a
    path $\steps{s_0}{s_1}{s_n}$ with $\phi(s_n)$ in
    $\ts$.
    By the invariant simulation property (Theorem \ref{thm:inv-run-simulation}) we get a path $\stepsle{t_0}{t_1}{t_n}$ with $s_n \tsle t_n$ and $t_n \in I$.
    Because $\phi$ is compatible, we have $\phi(t_n)$.
    However, this contradicts $t_n \in I$.
\end{proof}

We state this theorem again for $\tsle$ instead of $\tsse$ as we did above.
The proofs are the same.

\begin{theorem}
    \label{thm:inv-run-simulation}
    Let $I$ be a reachability invariant for $\ts$.
    Then, for all $s_1 \tsle t_1$ with $t_1 \in I$,
    if $\steps{s_1}{s_2}{s_n}$ is a path,
    then there is a path $\stepsle{t_1}{t_2}{t_n}$ such that $s_i \tsle t_i$
    and $t_i \in I$ for all $i$.
    Similarly, if $\run{s_1}{s_2}$ is a run,
    then there is a run $\runle{t_1}{t_2}$ such that $s_i \tsle t_i$ and $t_i \in I$
    for all $i$.
\end{theorem}
\begin{proof}
    By induction (coinduction) on the inductive (coinductive) definition
    of $\steps{s_1}{s_2}{s_n}$ ($\runle{s_1}{s_2}$).
\end{proof}

\begin{proposition}
    Let $f$ be a topological numbering of $\to$.
    If there exists a path of the form $\steps{s}{s_1 \to s_2}{s}$,
    then $\neg \phi(s)$.
\end{proposition}
\begin{proof}
    For the sake of contradiction, assume $\phi(s)$.
    From $s \to s_1$, we have $f(s) > f(s_1)$.
    Moreover, from $\steps{s_1}{s_2}{s}$, we have $f(s_1) \ge f(s)$
    by induction on the path definition.
    Together, we arrive at the contradiction $f(s) > f(s)$.
\end{proof}

\begin{proposition}
    \label{topo-numberings-complete}
    If there is no path $\steps{s}{s_1 \to s_2}{s}$ with $\phi(s)$ in $\to$,
    then the following are \fw s for $\to$.
    \begin{enumerate}
        \item The number of accepting states that are reachable from a node: $f(s) := |\{x \,|\, \reaches{s}{x} \wedge \phi(x)\}|$ (assuming $\{x \,|\, \reaches{s}{x} \wedge \phi(x)\}$ is finite for any $s$).
        \item If $h$ is a topological numbering (in the classical sense) of the strongly connected components (SCCs) of $\to$,
        then set $g(s) := h(C)$ if $s \in C$.
    \end{enumerate}
\end{proposition}
\begin{proof}
    \begin{enumerate}
        \item Clearly, if $s \to t$, then $f(s) \geq f(t)$. Suppose $s \to t$, $\phi(s)$,
        and $\{x \,|\, \reaches{s}{x} \wedge \phi(x)\} \subseteq \{x \,|\, \reaches{t}{x} \wedge \phi(x)\}$.
        Then, we have $\reaches{t}{s}$ and thus $\reachesp{s}{s}$. Because of $\phi(s)$, this contradicts
        the assumption and hence $f(s) > f(t)$.
        \item As $h$ is a topological numbering (in the classical sense, \ie\ without condition (2) of definition \ref{def:forward-numbering})
        of the SCCs of $\to$, we have $g(s) \geq g(t)$ if $s \to t$. Moreover, any state $s$ with $\phi(s)$ has to form its own trivial SCC as otherwise
        the SCC of $s$ would form an accepting cycle. Thus $g(s) > g(t)$ if $\phi(s)$ and $s \to t$.
    \end{enumerate}
\end{proof}

\begin{theorem}
    \label{thm:no-buechi-run}
    Let $f$ be a restricted \fw\ of $\ts$ for a compatible $\phi$ and a restricted reachability invariant $\rI$ such
    that $I$ is finite.
    Then, for any initial state $s_0 \in S$ such that there exists a $t_0 \in I$ with $s_0 \tsse t_0$,
    there is no B\"uchi run from $s_0$ in $\to$.
\end{theorem}
\begin{proof}
    Working towards a contradiction, suppose that there is an accepting B\"uchi run
    $\run{s_0}{s_1}$.
    By proposition \ref{thm:inv-run-simulation}, there exists a run $\run{t_0}{t_1}$ such that $s_i \tsse t_i$ and $t_i \in I$ for all $i$.
    As there are infinitely many $s_i$ with $\phi(s_i)$, and $\phi$ is compatible with $\ts$,
    there are also infinitely many $t_i$ with $\phi(t_i)$.
    Because $I$ is finite,
    accepting states have to repeat eventually and there is an accepting lasso,
    \ie\ there exists an accepting $t_k$ (with $\phi(t_k)$) such that $t_0 \to^* t_k$ and $t_k \to^+ t_k$.
    This contradicts proposition \ref{thm:topo-numbering-no-accepting-cycle}.
\end{proof}

\begin{theorem}
    \label{thm:certificate-checker-correct}
    If $\textsc{B\"uchi-Emptiness}(\phi, C, q_0)$ accepts the certificate,
    then $\dbmto$ has no B\"uchi run for $\phi$. Consequently, the underlying TBA is empty.
\end{theorem}
The proof will make clear why we needed to define the concept of a restricted reachability invariant
in the first place: to only select certain subsumptions that go downward as far as possible.
\begin{proof}
We define $I := \set{(q, Z) \,|\, \exists i.\, (q, Z, i) \in C)}$.
The whole proof hinges on the fact that we are able to define a suitable
$\tsse$ such that the certificate can form a restricted reachability invariant \mbox{$\rI$}
for $\ts$,
and in addition we can define a suitable numbering $f$ which is a restricted topological numbering
for $\ts$ with respect to $\tose$.
The main idea for this is to only consider subsumptions that go downward as far as possible.
That is, we have $(q, Z) \tsse (q', Z')$ iff $q' = q$, $Z \subseteq Z'$,
and if there exists a $k^*$ such that $(q, Z', k^*) \in C$ and $k^*$ is minimal among all
$k$ for which there exists a $Z''$ with $Z \subseteq Z''$ and $(q, Z'', k) \in C$.
Lines 7-10 of the algorithm ensure that there exists
one such triplet $(q, Z'', k)$ for each $(q, Z) \in I$.
Thus the first condition of definition \ref{def:restricted-reachability-invariant}
is verified. The condition $I \subseteq S$ is ensured by lines 2-4.
The second condition, finally, is trivially met by the definition of $\tsse$.

Now let $f(q, Z)$ be the minimal $i$ such that $(q, Z, i) \in C$.
We need to prove that the conditions of definition \ref{def:restricted-forward-numbering} are met.
Suppose $(q, Z) \in I$, $(q, Z) \zgto (q_1, Z_1)$ and $(q_1, Z_1) \tsse (q_1, Z_1')$.
We know that $(q, Z, f(q, Z)) \in C$ by the definition of $f$.
Lines 7-10 of the algorithm ensure that
there exists a triplet $(q_1, Z', k) \in C$ such that $Z \subseteq Z'$, $f(q, Z) \geq k$,
and $\phi(q) \implies f(q, Z) > k$.
Then, by the definition of $\tsse$, there exists a $k^*$ such that $(q_1, Z_1', k^*) \in C$
and $k \geq k^*$. Thus $k \geq f(q_1, Z_1')$ by the definition of $f$
and both conditions of \ref{def:restricted-forward-numbering} are met.

Finally, together with the fact that $(q_0, \{\vali\})$ is covered by $C$, which is ensured by lines 5-6, we can
invoke theorem \ref{thm:no-buechi-run} to show that $\to$ does not have a B\"uchi run.
Then, by completeness of $\zgto$
the underlying TBA is empty as well.
\end{proof}

\begin{proposition}
    \label{thm:zg-mono}
    Let $\tasim$ be a time-abstract simulation.
    If $\alphasim(W) \subseteq \alphasim(W')$, $(q, W) \zgto^t (q_1, W_1)$, and $(q, W') \zgto^t (q_1, W_1')$,
    then $\alphasim(W_1) \subseteq \alphasim(W_1')$.
\end{proposition}
\begin{proof}
    Suppose $v \in W_1$ and $u \tasim v$.
    We need to show that there exists a $v' \in W_1'$ such that $u \tasim v'$.
    By the definition of $\zgto^t$ there exists a $\delta$ and $u_0 \in W$ such that
    $\tastep{A}{(q, u_0)}{\delta}{t}{(q_1, v)}$. Because $\tasim$ is reflexive
    and $\alphasim(W) \subseteq \alphasim(W')$
    there is a $v_0 \in W'$ with $u_0 \tasim v_0$.
    As $\tasim$ is a time-abstract simulation, we can find a $v'$ such that
    $\tastep{A}{(q, v_0)}{\delta}{t}{(q_1, v')}$ and $v \tasim v'$.
    Thus we have $v' \in W_1'$ by definition of $\zgto^t$ and $u \tasim v'$ by transitivity.
\end{proof}

In comparison to the main part of the paper (Sec.\ \ref{sec:incorporating_state_of_the_art_abstraction}),
we will first prove a more abstract version of Thm.\ \ref{thm:topo-numbering-simulation}
and then instantiate it for the concrete case of timed automata.
\begin{theorem}
    \label{thm:topo-numbering-simulation}
    Assume that $\to$ simulates $\to'$ with $\siminv$ and that the following conditions hold:
    \begin{enumerate}
        \item $\tsse'$ simulates $\tsse$ with $\sim$.
        \item For all $s$, $s'$, $t$, and $t'$, if $s \sim s'$, $t \sim t'$ and $s' \tsse' t'$, then $s \tsse t$.
        \item For all $s$ and $s'$, if $s \sim s'$ then $\phi'(s') \implies \phi(s)$.
    \end{enumerate}
    Let $I' := \set{s' \,|\, \exists s \in I.\, s \sim s'}$
    and $f'(s') := \mathit{Min}\, \set{f(s) \,|\, s \in I \wedge s \sim s'}$.
    \begin{enumerate}
        \item Suppose that $\rI$ verifies condition (1) of definition \ref{def:restricted-reachability-invariant}
        for $\to$. Then $(I', \tsse')$ verifies conditions (1) of definition \ref{def:restricted-reachability-invariant} for $\to'$.
        \item Suppose that $f$ verifies conditions (1) and (2) of definition \ref{def:restricted-forward-numbering}
        for $\to$, $\tsse$, $I$, and $\phi$. Then $f'$ verifies conditions (1) and (2) of definition \ref{def:restricted-forward-numbering} for $\to'$, $\tsse'$, $I'$, and $\phi'$.
    \end{enumerate}
\end{theorem}
\begin{proof}
    \begin{enumerate}
        \item Assume $s \in I$, $s \sim s'$ and $s' \to' t'$.
        We need to show that there is an $r' \in I'$ with $t' \tsse' r'$.
        By simulation we can find a $t$ such that $s \to t$ and $t \sim t'$.
        Thus there is an $r \in I$ with $t \tsse r$.
        With (1), we know there is an $r'$ with $t' \tsse' r'$ and $r \sim r'$.
        \item Assume $s \in I$, $s \sim s'$, $s' \to' t'$, $t' \tsse' r'$,
        $r \in I$, and $r \sim r'$.
        There is an $s_0 \in I$ such that $f'(s') = f(s_0)$ and $s_0 \sim s'$.
        By simulation we can find a $t$ such that $s_0 \to t$ and $t \sim t'$.
        With (2), we have $t \tsse r$.
        Thus we have $f(s_0) \geq f(r)$ and $\phi(s_0) \implies f(s_0) > f(r)$.
        Moreover, $f'(r') \leq f(r)$ by definition of $f'$.
        Finally, with (3) and $f'(s') = f(s_0)$, we get $f'(s') \geq f'(r')$
        and $\phi'(s') \implies f'(s') > f(r')$.

    \end{enumerate}
\end{proof}

\emph{Instantiation of theorem \ref{thm:topo-numbering-simulation} for timed automata}.
As mentioned above, we instantiate the theorem for $\to \;:=\; \zgto$, $\to' \;:=\; \zgalphasim$,
and $(q, Z) \sim (q, W) \iff W = \alphasim(Z)$.
To satisfy condition (3), any $\phi$ which is compatible with $\alphasim$ and the instantiation $\phi' \;:=\; \phi$ suffice.
Let $\tsse^*$ be the subsumption relation defined in the proof of theorem \ref{thm:no-buechi-run}.
That is:
\begin{align*}
&(q, Z) \tsse^* (q, Z') \\
   \iff\ &\exists k.\, Z \subseteq \alphasim(Z')
   \wedge (q, Z', k) \in C \\
   &\ \wedge (\forall (q, Z'', k') \in C.\, Z \subseteq \alphasim(Z'') \implies k' \geq k)
\end{align*}
We set \footnote{In the statement of Thm.\ \ref{thm:topo-numbering-simulation} given in Sec.\ \ref{sec:incorporating_state_of_the_art_abstraction}, we set $\tsse \,:=\, \tsse^*$ for brevity.}:
\begin{itemize}
    \item $(q, Z) \tsse (q', W') \iff \exists Z'.\; W' = \alphasim(Z') \wedge (q, Z) \tsse^* (q', Z')$ and
    \item $(q, W) \tsse' (q', W')\\ \iff \exists Z, Z'.\; W = \alphasim(Z)\wedge W' = \alphasim(Z') \wedge (q, Z) \tsse^* (q', Z')$.
\end{itemize}
Then condition (1) is trivially satisfied.
Condition (2) can be verified by observing that $\tsse^*$ is constructed in a way such that it is deterministic,
\ie\ if $(q, Z) \tsse^* (q', Z_1)$ and $\alphasim(Z) = \alphasim(Z')$, then $(q, Z') \tsse^* (q', Z_1)$
(because $Z \subseteq \alphasim(Z_1) \iff Z' \subseteq \alphasim(Z_1)$
if $\alphasim(Z) = \alphasim(Z')$).
Finally, it follows from proposition \ref{thm:zg-mono} that $\zgto$ simulates $\zgalpha$ for $\siminv$
as outlined in Sec.\ \ref{sec:incorporating_state_of_the_art_abstraction}.

\end{document}